\documentclass[12pt,a4,twoside]{amsart}
\usepackage{amssymb}
\usepackage{latexsym}

\usepackage{epsfig}
\usepackage{srcltx}
\usepackage{hyperref}
\usepackage{pdfpages}

\usepackage[hang]{subfigure}

\theoremstyle{plain}

\newtheorem{lemma}{Lemma}
\newtheorem{corollary}{Corollary}
\newtheorem{proposition}{Proposition}

\theoremstyle{definition}

\newtheorem{remark}{Remark}

\theoremstyle{plain}
\newtoks\thehProclaim
\newtheorem*{Proclaim}{\the\thehProclaim}

\theoremstyle{definition}
\newtoks{\thehRemark}
\newtheorem*{Remark}{\the\thehRemark}

\begin{document}

\title{Hooke and Coulomb energy of tripod spiders}

\author{Giorgi Khimshiashvili}
\address{Ilia State University, Tbilisi, Georgia}

\author{Dirk Siersma}
\address{Mathematisch Instituut, Universiteit Utrecht}



\begin{abstract}
Tripod spiders are the simplest examples of arachnoid mechanisms. Their workspaces and configuration spaces are well known.
For Hooke potential, we give a complete description of the Morse theory and treat the robust control of the spider. For the Coulomb energy,  we use  stationary charges and the trapping domain to study the robust control of spiders.  We show that, for a regular triangle and positive charges, the domain of robust control is non-void. This relates to questions about the Maxwell conjecture about point charges. We end with several natural problems and research perspectives suggested by our results.
\end{abstract}

\date{\today}

\maketitle

\section{Introduction}

The  geometry and topology of  mechanical linkages play an important and increasing role in applied problems.  
Most of the previous studies were concerned with the workspace and the topology of the configuration space, 
which is only known in a few cases summarized in \cite{KM}, \cite{Oh}, \cite{Mo}. The aim of this paper is 
to extend the known results to new classes of linkages and enrich them by considering potential functions on
the workspace, with aplications to control of the linkages considered. Our approach is based on Morse theory 
which yields an explicit connection between the topology of configuration space and the critical points of potential.
 
We will be basically concerned with the so-called arachnoid mechanisms, the topology of which is a largely unexplored topic. 
Nowadays the same type of objects is often called "spidery linkage". Detailed information on the topology of arachnoid 
linkages is important for the design and control of certain types of spider robots. More concretely, we study the simplest 
arachnoid mechanism, the $3$-leg spider also known as tripod spider \cite{Oh}, \cite{Mo}. 

In Section \ref{s:hooke} we  use Hooke energy 
as a potential function and describe in detail the workspace and the critical point theory. This implies, in particular, that 
a weighted version of Hooke energy can be used to control such a linkage.
 
In section \ref{s:coulomb} we study the Coulomb potential of point charges placed at the foots of a $3$-leg spider linkage. We deal with the so-called 
stationary charges of the spider's center (the common point of legs)  and the so-called trapping domain of stationary charges 
determined in \cite{GK1},  \cite{GK2}. This enables us to determine the domain in the workspace of spider, where the position of
spider's center can be robustly controlled by the values of stationary charges using the Coulomb control scenario developed in \cite{KPS}, 
\cite{GK1}. For a symmetric spider based on a regular triangle and having a contractible workspace, we show that the domain of 
robust Coulomb control is a non-void open subset of the workspace containing the center of the reference triangle. 

In the sequel we use Morse theory for manifolds with boundaries and corners. Most of it is "folklore" hidden in the literature. 
We mention here \cite{JJT} and \cite{GM} as the most general reference. For critical points of functions on manifolds with 
boundary and corners we refer to \cite{Si}. For brevity, we refer to criteria of such critical points and corresponding 
topological changes in level surfaces as "standard rules".

In conclusion we mention several related problems and perspectives suggested by our results. In general, this paper may 
be considered as a first step in applying our approach to spider linkages and creating a paradigm for further research 
in this direction. 

The results of this paper were obtained and written up in the framework of a "Research in Residence" project at 
the "Centre des Rencontres Mathematiques" (CIRM, Luminy, France) realized in November of 2022. It is our pleasure 
to acknowledge the support and excellent working conditions at CIRM which largely facilitated our research.

\section{Hooke Energy as a potential function of tripod spider}\label{s:hooke}

We will consider the Hooke  potential for three points  in  several situations.  First without any constraint, next with constraints on maximal an minimal distances, and finally  for a 3-leg spider.

\subsection{No constraints}
We start with the following simplified situation. Given $3$ points $A,B,C$ (in vector notation $a,b,c$), the foot points  of the  spider, 
and its center (joint) the  point $X$. The legs $AX$, $BX$, $CX$ are completely flexible around the foot, their length is allowed to change. 

\begin{center}
\includegraphics[width=6cm]{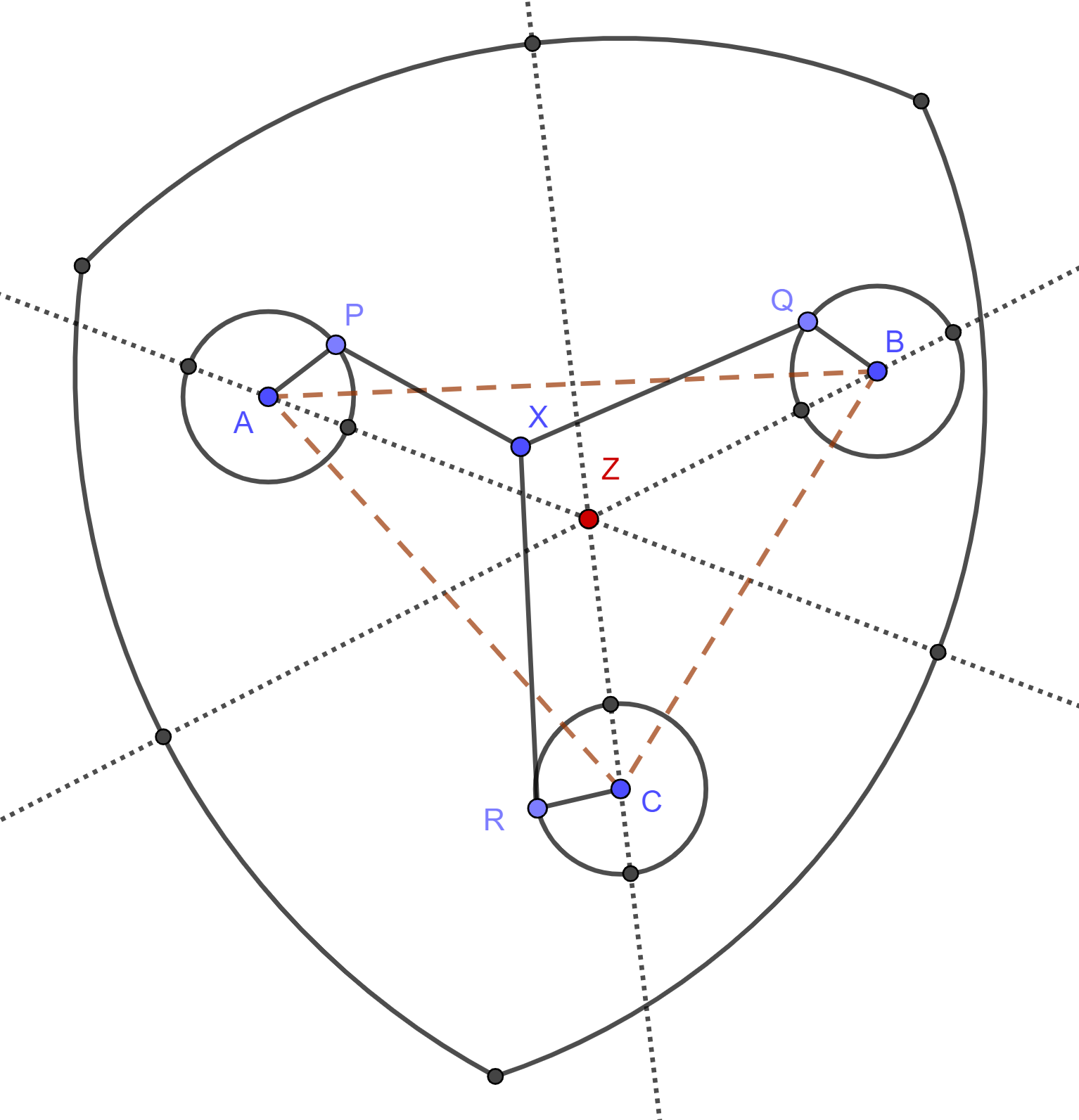}
\end{center}

\smallskip
\noindent
The Hooke Energy  is defined by
$$  H(x) =  ||x- a||^2  +||x-b||^2 +  ||x- c||^2$$
The stationary points are determined by: 
$\nabla H = 6 x - 2(a+b+c) = 0$, so $x = \frac{1}{3} (a+b+c)$, the center of gravity.
There are no other stationary points.
Note that $H(X) = ||x-z||^2 +  K,$ where $z=\frac{1}{3} (a+b+c)$, the center of gravity $Z$ and $K= ||a||^2 + ||b^||^2+||c||^2 - ||z||^2$. 
All level curves are circles.

\subsection{Maximal length constraints}
We require
$$  |XA| \le R_A \; , \; |XB| \le R_B \;  , \; |XC|  \le  R_C.$$
The configuration space is the intersection of 3 discs $D_A, D_B ,D_C$ with centers resp $A,B,C$  and radii $R_A,R_B,R_C$. We study the boundary extrema of $H$.

\begin{lemma}
The map $H$ restricted to the boundary of $D_A$ has an extrema in $Y$ iff $AY$  has the direction of $AB + AC$,  i.e the  two intersection points of the line trough $A$ and $Z$ with the boundary of $D_A$, a minimum and a maximum.

\end{lemma}
\begin{proof}
Suppose $A$ is the origin.
Let $R_A e_{\phi}$ a point on $\delta D_A$. Use Lagrange multiplyers:
$$ \nabla H = 6 R_A e_{\phi} - 2 (b+c) = \lambda e_{\phi}$$
Therefore
$$ (- \lambda + 6 r_A) e_{\phi} = 2(b+c)$$
\end{proof}

\begin{proposition}
 Let  the workspace $W = D_A \cap D_B \cap D_C$ contain an open neighborhood of $Z$ then $Z$ is the only stationary point of $H$, an absolute minimum. There are no boundary singularities.
\end{proposition}
\begin{proof}
The center of gravity $Z$ is clearly a  minimum.
Potential other critical points are the critical points of the restiction of $H$ to the boundary circles (intersections of $ZA$ with the circle around $A$, etc) and also the 3 `corner points `, where 2 boundary circles intersect. The statement follows now from the standaard rules . 
For boundary points, they are here  as follows:

\smallskip
\noindent
\begin{tabular}{llll}
type of restriction & direction normal & type & cell attaching \\
\hline
minimum    & inward  &  no change & no\\
minimum  & outward  & minimum & 0-cell\\
maximum    & inward  &  no change & no\\
maximum & outward  &  saddle & 1-cell \\
\end{tabular}

\smallskip
\noindent
With no change we mean, that the topological type of the lower level sets don't change.\\
For corner points there are similar rules. They give in our case: no change. Due to the special (circular) form of level curves and boundaries one can obtain the same result  by `inspection of pictures'.
\end{proof}

\subsection{Two sided length constraints}
We require
$$  0 < R_A^{-} \le |XA| \le R_A \; , \; 0 < R_B^{-} \le |XB| \le R_B \;  , \; 0 < R_C^{-} \le |XC|  < R_C .$$
The configuration space is the intersection of 3 discs with centers $A,B,C$  and radii $R_A,R_B,R_C$, where some smaller (open)  discs have been taken out. From the  many different posibilties, we consider here the case that these 3 discs have a relatively small radii and the workspace $W$ is a disc with 3 holes, containing an open neighborhood of the gravity center  $Z$. 

\begin{proposition}\label{p:3holes}
 In this situation the point  $Z$ is an absolute minumum of $H$ on $W$  , moreover there are 3 saddle points on the (outer) intersection points of  the small discs with $ZA$, $ZB$, $ZC$.
\end{proposition}
\begin{proof}
The center of gravity $Z$ is an absolute minimum.
Other potential  critical points are the critical points of the restiction of $H$ to the boundary circle,  the 3 small circles and also the 3 `corner points `, where 2 boundary circles intersect.
The statement follows now from standard rules for boundary and corner singularities as explained above.

\end{proof}
\begin{remark}
The phase portrait  of the gradient of $H$  consists of  straight lines to $Z$, except for those  lines that intersect one of the small discs.
Their trajectories follow from the moment that they intersect these discs  a part of the boundary circle until they become `visible' from $Z$ and then continue as a straight line. There are 3 intervals between the outer and the inner circles, which are conflict strata for the gradient flow.
\end{remark}

\begin{remark}
The Morse theory in the 3 cases above is as follows:\\
In 2.1:  $b_0 =1, b_i = 0 \; (i \ge 1)$  and  $\mu_1=1 ,  \mu_i = 0  \;  (i \ge 1)$,\\
In 2.2:  $b_0 =1, b_i = 0 \; (i \ge 1)$  and  $\mu_1=1 , \mu_i = 0  \;  (i \ge 1)$,\\
In 2.3:  $b_0 =1, b_2= 3  ,  b_i = 0  \; (i \ge 2)$  and  $\mu_1=1, \mu_1 = 3 , \mu_i = 0  \; (i \ge 2)$.\\
In these three cases, $H$ is a perfect Morse function.
\end{remark}

\subsection{Robust control}
A similar study can be made for the weighted Hooke Energy:
$$ H_{\alpha,\beta,\gamma} =  \alpha ||x- a||^2  + \beta ||x-b||^2 + \gamma  ||x- c||^2$$
Assume  $\alpha > 0,  \beta >0,  \gamma > 0$. This potential function has the point $ Z = (\alpha : \beta : \gamma)  $ (barycentric coordinates) as an absolute minimum.  The level curves are circles with the center  at this point $Z$. The critical point theory is similar to the case of $H$, which corresponds to  $ (\alpha: \beta :\gamma ) = (1:1:1)$.

\smallskip
\noindent
A proper choice of the controls $(\alpha,\beta,\gamma)$ can be used to move the 3-leg spider  to any point in the triangle $ABC$ via minimum points of Hooke energy. This procedure yields a \textit{robust control} of the spider.

\subsection{The 3-leg spider}
In this case the telescopic connections are replaced by 2-arms with fixed arm lengths and flexible turning point.
We assume that the two parts of the arm have different lengths.
So $AX$ is replaced by the arm $AP\cdot PX$,  $BX$ by $BQ \cdot QX$ and $CX$ by $CR \cdot RX$. The configuration space $\mathcal C$ is an 8-fold cover of the workspace $W$ with certain identifications at the boundaries.  
The topology has been studied in full generality by P. Mounod \cite{Mo}. He showed that an $n$-leg spider with generic arm lengths has a smooth two-dimensional configuration space, and gave a formula for its Euler characteristic.
In his paper he also solved some questions of J. O'hara \cite{Oh}.

The configuration space of 3-leg spider projects now clearly to the workspace $W$ from 2.3 (with the 3 small holes). 
$R_A^-$ is equal to the difference of arm-lengths in the  $AP\cdot PX$ and $R_A$ is equal to their sum. Similar for the points $B$ and $C$.
We assume here that the arm lengths are  such, that they give the workspace mentioned above.

\smallskip
\noindent
 According to the formula of \cite{Mo} the Euler characteristic of $\mathcal{C}$ is -22. So we have a smooth  surface with genus 12.

We consider the quadratic distance function $H$ (see above).
\begin{proposition}
The Hooke energy $H$ on the configuration space $\mathcal{C}$ of the 3-leg spider has:
\begin{itemize}
\item 8  minima
\item 36 saddles
\item 6 maxima
\end{itemize}

\end{proposition}
\begin{proof}
The potential $H$ is defined on the configuration space via projection to the work space. We will use the results from Propostion  \ref{p:3holes} and
use (branched) covering arguments. The covering is 8-fold above the open part of the workspace. The branching takes place at the boundaries and corners. The potential critical points are just the  preimages of the special points, which have been studied in Proposition \ref{p:3holes}.

We find in this way $8$ critical pre-images of the center of gravity $Z$, these are all absolute minima (with the same value).
The $36$ saddles come from the $9$ special points on boundary circles, where we have covering degree $4$. Their types can be deduced
by topological study of the gluing of the two boundary pieces. Finally, we have $6$ maxima, which corespond to the three corner points, 
where the covering is $2$-fold and at each such point such four pieces are glued together.

There is is a dictionary between the Morse theory on the workspace treated above and the configuration space. Compare the level curves 
in both cases via gluing and this results in the Morse indices  mentioned in the  proposition.
\end{proof}

\section{Coulomb Energy as a potential function  of tripod spider} \label{s:coulomb}
From the viewpoint of control theory, it is also interesting and practically important to consider the Coulomb potential of point charges which are placed at the fixed foot points of a spider and can be varied in order to change the position of its body endowed with a fixed charge. A similar scenario was studied in big detail in the papers \cite{GK1}, \cite{GK2} and our exposition in this section relies on the constructions and results given in those two papers.

\smallskip
\noindent
We begin with recalling several concepts and constructions  in the form needed in the sequel. Recall that the Coulomb potential $E = E(Q@A)$ of $n$ point charges of non-zero magnitudes $q_i$ placed at $n$ points $A_i$ is a function of point $X$ in the complement of points $A_i$ by the formula
$$E = \sum \frac{q_i}{d_i}, $$
where $d_i = d(X, A_i)$ is the Euclidean distance between the points $P$ and $A_i$. As usual a point $X$ is called a stationary point (or equilibrium point) of potential $E$ if its gradient $\nabla E$ vanishes at $X$. There exists a huge number of results and problems concerned with the Coulomb potential and equilibrium points of point charges. For us it is important to mention that, as was proven by M. Morse himself, for a generic configuration and generic values of point charges Coulomb potential is a Morse function \cite{CM}. This suggests that Coulomb potential is a reasonable candidate for developing Morse theory for a spider linkage along the same lines as in the case of Hooke’s energy considered above. The aim of this section is to present a number of results in this direction and present an application in the spirit of robust Coulomb control discussed in \cite{KPS}, \cite{GK1}, \cite{GK2}.   
Another relevant topic in our context is the so-called Maxwell’s conjecture on point charges stating that the number of isolated equilibrium points of $n$ point charges in $2$-dimensional Euclidean space does not exceed $(n-1)^2$ (see, e.g., a recent review in  \cite{GaNoSh}). 

\smallskip
\noindent
We next focus on the case of 3 point charges. Given a triangle $\triangle$ with vertices $A, B, C$, referred to as a reference triangle (or base triangle), 
and a point $X$ in the same plane, the triple of normalized stationary charges $Q(X) =  \{q_i(X; \triangle )\}$ is defined  as a triple of non-zero real 
numbers $q_i$ such that $\nabla E(X,Q) = 0$. As was proven in \cite{GK1} these charges are uniquely determined by the ratio's
$$q_1:q_2:q_3 = d_1^3 \mathcal{A}_1 : d_2^3 \mathcal{A}_2 : d_3^2 \mathcal{A}_3,$$
where $\mathcal{A}_{1}$ is the area of $\triangle BCX$, $\mathcal{A}_2$ of $\triangle CAX$, $\mathcal{A}_3$ of $\triangle ABX$. We denote the angles at X of these triangles by $\alpha_1,\alpha_2,\alpha_3$ respectively.

\smallskip
\noindent
 In \cite{GK1}  , \cite{GK2}  \textit{ the Coulomb trapping domain} $T(\triangle)$  is defined as the set 
of  minimum points of the potential induced by the charges $Q(X)$. The trapping domain is given by the inequality $h(X) > 0$, where $h$ is the Hessian of $E(X,Q(X))$, which we call the \textit{ trapping Hessian}. We have  in geometric terms the following formula:
$$ h(X) =  -2 \mathcal{A} +9 (\prod_{i}^3  \sin \alpha_i  )(\sum_{i=1}^3 d_i^2 \mathcal{A}_i).$$

It was conjectured in \cite{GK2} that 
   $T(\triangle)$ is a non-empty subset containing the incenter of  $\triangle$. This conjecture has been verified in a number of cases, including
the regular triangle. 

\smallskip
\noindent
As was mentioned, a complete description of equilibria (stationary points) of $E$ and their Morse indices is not known even in the case of three charges. 
However rather complete results for special configurations of three point charges have been obtained in \cite{Ts}. In particular, the Maxwell’s conjecture 
was proven for arbitrary three non-zero point charges at the vertices of a regular triangle and the Morse indices are computed. 

\smallskip
\noindent
For this reason, in the rest of the text we always assume that the reference triangle is regular since it is the most important case for applications 
and, at the same time, in this case all necessary background results are rigorously proven in \cite{GK2}.

\smallskip
\noindent
Further, given a $3$-leg spider $S$ with the center at (moving) point $X$, (fixed) foot-points $A, B, C$ and identical legs with the links of lengths $a$ 
(thigh length) and $b$ (foot length), we denote by $W(S)$ its workspace considered above.  To simplify the discussion we assume that $R >a > b > 0$, 
in which case the whole workspace of the spider is the intersection of three  circular annuli with centers $A, B, C$.

\begin{proposition}
The domain of robust Coulomb control of spider $S$ in the above scenario is equal to the intersection $D(S) = W(S) \cap T(\triangle)$.
\end{proposition}

\noindent
As said before the topology of the workspace depends on the triangle and the arm lengths.  
This influences the intersection which could even be void. For the regular triangle, we make the choice that the workspace becomes a contractible 
region bounded by three circular arcs and containing the incenter.  Also we assume that all charges are positive. 
The trapping domain $T(\triangle)$ is explicitly known as the interior of a compact convex region bounded by $h=0$ and containing
the incenter of the triangle. Under these conditions this yields the following conclusion.
\begin{corollary}
The domain of robust Coulomb control of spider $S$ based on a regular triangle is non-void.

\end{corollary}

\medskip
\noindent
For a spider based on the regular triangle one can also carry out a complete Morse analysis in a similar way as for the Hooke potential.
One knows that for given positive charges there is a single minimum of E in the trapping area $T(\triangle)$ and $3$ saddle points 
in $\triangle$  ouside  $T(\triangle)$.  The poles of $E$ are not in the workspace and therefore not in the configuration space, which 
is an $8$-fold cover of the workspace branched (glued) over boundary components. After obtaining information about the boundary singularities 
one gets a complete Morse analysis.

\section{Concluding remarks}

There are several natural problems and research perspectives suggested by our results. We mention
some of them which seem most interesting and feasible.

First of all, note that all the concepts used in our paper make sense for $n$-leg spiders with arbitrary base $n$-gon 
and arbitrary lengths of the links. So a natural next step is to generalize the above results to Hooke and Coulomb 
potential of an $n$-leg spider.

Next, one can also consider other geometrically or physically meaningful potentials of $n$-leg spider like Riesz energies or 
the oriented area of the polygon formed by the moving joints of spider.

Further, one can formulate several extremal problems related to the design of spider linkages with desirable properties
like the area or shape of the workspace.

Moreover, one can try to extend our results to the case of $n$-leg spiders with a moving $n$-gon platform instead of the center
point. 

Finally, it is natural to extend our discussion to spatial spiders where the legs can move in a three-dimensional Euclidean space.

It is also possible to consider "legs" with more than two links or with flexible non-extendible tether links. Each of the mentioned
possibilities deserves a closer look and careful exploration which we intend to undertake in the future research.  

Returning to $3$-leg spiders we add that an urgent topic is to clarify what happens with the domain of robust Coulomb control for 
arbitrary configuration of foots. In this case, there are several feasible problems concerned with the optimization of workspace with 
given restrictions on the sum of the lengths of links and circumradius of the base triangle. For example, there is good evidence 
that the maximal area of $W(S)$ arises for the most symmetric spider with equal links and regular configuration of foots.



\end{document}